\documentclass{article}
\usepackage{amsmath,amsfonts,amsthm,amssymb}
\usepackage{setspace}
\usepackage{fancyhdr}
\usepackage{lastpage}
\usepackage{extramarks}
\usepackage{chngpage}
\usepackage{soul}
\usepackage{graphicx,float,wrapfig}
\usepackage{fullpage}
\usepackage[linesnumbered,boxed,ruled,vlined]{algorithm2e}
\usepackage{bm}
\usepackage{soul}
\usepackage[usenames,dvipsnames]{color}
\usepackage{calc}

\newsavebox\CBox
\newcommand\hcancel[2][0.5pt]{%
  \ifmmode\sbox\CBox{$#2$}\else\sbox\CBox{#2}\fi%
  \makebox[0pt][l]{\usebox\CBox}%
  \rule[0.5\ht\CBox-#1/2]{\wd\CBox}{#1}}

\newcommand{\eat}[1]{}
\newcommand{\supp}{\mathsf{supp}}
\newcommand{\oracle}{\mathcal{O}}
\newcommand{\poly}{\mathrm{poly}}
\newcommand{\tF}{\widetilde{F}}
\newcommand{\q}{\Delta}
\newcommand{\calT}{\mathcal{T}}
\newcommand{\calL}{\mathcal{L}}
\newcommand{\calA}{\mathcal{A}}
\newcommand{\calR}{\mathcal{R}}
\newcommand{\sfL}{\mathsf{L}}
\newcommand{\N}{\mathbb{N}}
\newcommand{\Sg}{\mathsf{Seg}}
\newcommand{\hT}{\widehat{T}}
\newcommand{\hP}{\widehat{P}}
\newcommand{\T}{\mathcal{T}}
\newcommand{\K}{Z}
\newcommand{\hatPr}{\widehat{\Pr}}
\newcommand{\hK}{\widehat{Z}}
\newcommand{\Z}{\mathbb{Z}}
\newcommand{\R}{\mathbb{R}}

\newcommand{\I}{\mathbb{I}}
\newcommand{\eps}{\epsilon}

\newcommand{\strin}[1]{{#1}}

\newcommand{\done}[1]{{#1}}

\newtheorem{theorem}{Theorem}[section]
\newtheorem{lemma}[theorem]{Lemma}

\begin{document}

\setcounter{page}{1}

\title{
A Fully Polynomial-Time Approximation Scheme for \\ Approximating a Sum of Random Variables
\thanks{Institute for Interdisciplinary Information Sciences,
Tsinghua University, China.} %
}
\author{ 
Jian Li
\thanks{
Email: lijian83@mail.tsinghua.edu.cn
}
\quad\quad\quad\quad Tianlin Shi
\thanks{
Email: stl501@gmail.com
}
}
\date{}

\maketitle

\begin{abstract}
Given $n$ independent integer-valued random variables $X_1, X_2, ..., X_n$ and an integer $C$, we study
the fundamental problem of computing the probability that the sum $\bm{X}=X_1+X_2+...+X_n$ is at most $C$.
We assume that each random variable $X_i$ is implicitly given by an oracle $\mathcal{O}_i$, which given two input integers $n_1, n_2$
returns the probability of $n_1 \leq X_i \leq n_2$.
We give the first deterministic fully polynomial-time approximation scheme (FPTAS)
to estimate the probability up to a relative error of $1\pm \epsilon$.
Our algorithm is based on the technique for approximately counting knapsack solutions, developed in [Gopalan et al. FOCS11].
\end{abstract}

\begin{spacing}{1.1}

\section{Introduction}
We study the following fundamental problem.
The input consists of $n$ independent
(not necessarily identically distributed)
random integral variables $X_1,\ldots, X_n$ and an integer $C$.
Our task is to compute the following probability value 
\begin{equation}\label{eq:countprob}
F(C)=\Pr\,\Bigl[\,\sum_{i=1}^{n}{X_i} \leq C\,\Bigr]
\end{equation}

It is well known that computing $F(C)$ is \#P-hard (see e.g., \cite{kleinberg1997allocating}).
The hardness of computing $F(C)$ has an essential impact in the area of stochastic optimization
as many problems generalize and/or utilize this basic problem in one way or another,
thus inheriting the \#P-hardness.
Although we can sometimes use for example the linearity of expectation to bypass the difficulty of computing
$F(C)$, more than often no such simple trick is applicable,
especially in the context of risk-aware stochastic optimization
where people usually pay more attention to the tail probability than the expectation.

Despite the importance of the problem,
surprisingly, no approximation algorithm with provable multiplicative factor is known.
We note that we can easily obtain an additive PRAS (polynomial-time randomized approximation scheme)
for this problem via the Monte-Carlo method:
for each $i\in \{1,2,...,n\}$, generate $K$ independent samples $X_i^{(k)}, k = 1,2,...,K,$
according to the distribution of $X_i$, and then use the empirical average
\begin{equation*}
\tF(C) = \frac{1}{K}\sum_{k=1}^K{\I\,\Bigl(\sum_{i = 1}^n X_i^{(k)} \leq C\Bigr)}
\end{equation*}
as the estimation of $F(C)$, where $\I(\cdot)$ is the indicator function.
It is easy to see that $\tF(C)$ is an unbiased estimator of $F(C)$.
By standard Chernoff bound, one can see
that with $K=\poly(1/\epsilon)$ samples, the estimate is within an additive error $\epsilon$ from the true value
with constant probability (see e.g., \cite{mitzenmacher2005probability}).
To get a reasonable multiplicative approximation factor (say a constant close to 1), we need to set the additive error
 at the order of $F(C)$. So the number of samples needs to be
$\poly(1/F(C))$, which can be exponentially large, when $F(C)$ is exponentially small
\footnote{
In certain application domains such as risk analysis,
small probabilities (often associated with catastrophic losses) can be very important.
}.

\noindent {\bf Assumptions.}
Before presenting our main result, we need some notations and assumptions of the computation model.
We assume that all random variables are discrete and the support of $X_i$, denoted as $\supp_i$,
is finite and consists of only integers.
Without loss of generality, we can assume all $X_i$s are nonnegative (i.e., $\supp_i\subseteq \mathbb{N}$)
and $0\in \supp_i$ for all $i$.
To see why this is without loss of generality, simply consider the equivalent problem
of computing $\Pr[\sum_{i=1}^n (X_i-\min X_i) \leq C-\sum_{i=1}^n \min X_i  ]$,
where $\min X_i$ is the minimum value in $\supp_i$. Under such an assumption, the problem is non-trivial only for $C > 0$.
Moreover, we can assume that $\supp_i\in [0,C+1]$ for all $i$
since we can place all mass in $[C+1,\infty)$ at the single point $C+1$, which does not affect the answer.
The distribution of each random variable $X_i$ is implicitly given by an oracle
$\oracle_i$, which given two input value $(n_1,n_2)$
returns the value $\Pr[n_1\leq X_i\leq n_2]$ in constant time.

Our main result is a \emph{fully polynomial-time approximation scheme} (FPTAS) for computing $F(C)$.
For ease of notation, we use $(1\pm\epsilon)F(C)$ to denote the interval $[(1-\epsilon)F(C), (1+\epsilon)F(C)]$.
Let $\q=\prod_i \Pr[X_i=0]$.
Clearly $\q$ is a lower bound on the solution.
Recall that we say there is an FPTAS for the problem, if for any positive constant $\epsilon>0$,
there is an algorithm which can produce an estimate $\tF$ with
$\tF\in (1\pm\epsilon)F(C)$ in
$\poly(n, \epsilon^{-1}, \log C, \log \frac{1}{\q})$ time
\footnote{
More precisely, such a running time is called {\em weakly} polynomial time
since it is polynomial in the $\log$ of the numerical values of the instance
(or the number of bits required to encode those numerical values).
On the contrary, a strongly polynomial time for this problem would be a polynomial only in $n$ and $1/\epsilon$.
} (See e.g., \cite{mitzenmacher2005probability}).

\begin{theorem}
\label{th:main}
We are given $n$ independent nonnegative integer-valued random variables $X_1,\ldots, X_n$,
a positive integer $C$, and
a constant $\epsilon > 0$.
Suppose that for all $i \in \{1,2,...,n\}$, $\supp_i\subseteq [0,C+1]$, $0\in \supp_i$ and
there is an oracle $\oracle_i$, which, upon two input integers $(n_1,n_2)$,
returns the value $\Pr[n_1 \leq X_i \leq n_2]$ in constant time.
There is an FPTAS for estimating
$\Pr[\sum_{i=1}^n X_i\leq C]$ and the running time is
$O\bigl(\frac{n^3}{\epsilon^2} \log(\frac{1}{\q})^2 \log C\bigr)$.
\end{theorem}

\noindent
{\bf Remark 1.}
For simplicity of presentation, we assume in the above theorem a computation model in which any real arithmetic
can be performed with perfect accuracy in constant time and the probability values returned by the oracle
are reals, also with perfect accuracy. In Section~\ref{sec:bit}, we show how to implement our algorithm
in a computation model where only bit operations are allowed and the oracles also return numerical values
with finite precision. We show that the bit complexity of the algorithm is still a (somewhat larger) polynomial.

\vspace{0.2cm}
\noindent
{\bf Remark 2.}
Note that the oracle assumption is weaker than assuming the explicit representations of the distributions
(i.e., listing the probability mass at every point). In fact, if the input is the explicit representations of the distributions,
we can preprocess the input in linear time so that each oracle call to $\oracle_i$ can be simulated in $O(\log |\supp_i|)$ time.
This can be done by computing the prefix sums $\Pr[X_i\leq x]$ for all $x\in \supp_i$ in $O(|\supp_i|)$ time.
Then for each oracle call $(n_1,n_2)$, we use binary search to find out the smallest value $x_1 \in \supp_i$
that is no smaller than $n_1$ and
the largest value $x_2 \in \supp_i$
that is no larger than $n_2$ in $O(\log |\supp_i|)$ time.
Therefore, $\Pr[n_1\leq X_i\leq n_2]$ is the same as $\Pr[x_1\leq X_i \leq x_2]$, which can be computed from the prefix sums in constant time.

\subsection{Related Work}





There is a large body of work on estimating or upper/lower-bounding the distribution of the
sum of independent random variables. See e.g., \cite{bennett1962probability, petrov1965probabilities, lugannani1980saddle, daniels1987tail, mehta2007approximating}.
Those works are based on analytic numerical methods (e.g., Edgeworth expansion, saddle point method)
which either require specific families of distributions and/or do not provide any provable
multiplicative approximation guarantees.

Our problem is a generalization of the counting knapsack problem.
For the counting knapsack problem, Morris and Sinclair~\cite{morris2004random} obtained the first FPRAS (fully polynomial-time randomized approximation scheme)
based on the Markov Chain Monte-Carlo (MCMC) method.
Dyer~\cite{dyer2003approximate} provided a completely different FPRAS based on dynamic programming.
The first deterministic FPTAS is obtained by Gopalan et al.\cite{gopalan2011fptas}
(see also the journal version \cite{vstefankovic2012deterministic}).

Our problem is also closely related to the threshold probability maximization problem
(see a general formulation in \cite{li2013stoch}).
In this problem, we are given a ground set of items.
Each feasible solution to the problem
is a subset of the elements satisfying some property
(this includes problems such as shortest path,
minimum spanning tree, and minimum weight matching).
Each element $b$ is associated with a random weight $X_{b}$.
Our goal is to to find a feasible set $S$ such that $\Pr[\sum_{b\in S}X_b\leq C]$
is maximized, for a given threshold $C$.
There is a large body of literature on the threshold probability maximization problem,
especially for specific combinatorial problems and/or special distributions.
For example, Nikolova, Kelner, Brand and Mitzenmacher~\cite{nikolova2006stochastic} studied the corresponding
shortest path version for Gaussian, Poisson and exponential distributions.
Nikolova~\cite{nikolova2010approximation} extended this result to an FPTAS
for any problem with Gaussian distributions, if the deterministic version of the problem
has a polynomial-time algorithm.
The minimum spanning tree version with Gaussian distributed edges has also been studied in~\cite{geetha1993stochastic}.
For general discrete distributions, Li and Deshpande~\cite{li2011maximizing} obtained an additive PTAS
if the deterministic version of the problem can be solved exactly in pseudopolynomial time.
Very recently, Li and Yuan \cite{li2013stoch} further generalized this result to the class of problems
for which the multi-objective deterministic version admits a PTAS.

Our problem is also closely related to the fixed set version of the stochastic knapsack problem.
In this problem, we are given a knapsack of capacity $C$ and a set of items with random sizes and profits.
Their goal is to find a set of items with maximum total profit subject to the constraint that
the overflow probability is at most a given parameter $\gamma$.
Kleinberg, Rabani and Tardos \cite{kleinberg1997allocating} first considered the problem with Bernoulli-type distributions
and provided a polynomial-time $O(\log 1/\gamma)$-approximation.
Better results are known for specific distributions, such as exponentially distributions~\cite{goel1999stochastic}, Gaussian distributions~\cite{goyal2009chance,nikolova2010approximation}.
For general discrete distributions, bi-criteria additive PTASes
\footnote{
The overflow probability constraint may be violated by an additive factor $\epsilon$ for any constant $\epsilon>0$.
}
are known via different techniques \cite{bhalgat10,li2011maximizing,li2013stoch}.

\section{Algorithm}

Our algorithm is based on dynamic programming.
In Section~\ref{subsec:dp}, we provide
the recursion of the dynamic program, which is largely based on the idea developed in \cite{gopalan2011fptas,vstefankovic2012deterministic},
with some necessary adaptations.
However, since the support of each random variable can be exponentially large,
it is not immediately clear how the recursion can be implemented efficiently given the oracles.
We address this issue in Section~\ref{subsec:imp}.
In Section~\ref{sec:bit}, we analyze the bit complexity of our algorithm.

\subsection{The Dynamic Program}
\label{subsec:dp}
We first notice that
$\Pr[\sum_{j = 1}^{i}{X_j} \leq C]$, for any $i \in \{1,2,...,n\}$, is a nondecreasing function of $C$.
We consider its inverse function
$\tau(i,a): \{1,2,...,n\} \times \mathbb{R}_{\geq 0} \rightarrow \N \cup \{\pm \infty\}$,
which is defined to
$$
\tau(i,a)=\left\{
            \begin{array}{ll}
              \min\{C\mid C\geq 0 \text{ and }\Pr[\sum_{j = 1}^{i}{X_j} \leq C]\geq a\}, & \hbox{$0<a\leq 1$;} \\
              +\infty, & \hbox{$a>1$;} \\
              -\infty, & \hbox{$a=0$.}
            \end{array}
          \right.
$$
It is easy to see that $\tau(i,a)$ is nondecreasing in $a$.
The following simple lemma is needed. We omit the proof, which is straightforward.

\begin{lemma}
Both of the following statements hold true.
\begin{enumerate}
\item $\Pr\Bigl[\,\sum_{i=1}^{n}{X_i} \leq C\,\Bigr] = \max\{a : \tau(n,a) \leq C\}.$
\item $\tau(i,a)=0$ for any $i\in \{1,2,...,n\}$ and $a\leq \q$,
where $\q=\prod_i \Pr[X_i=0]$.
\end{enumerate}
\end{lemma}

The following recursion is very important to us.
For $x\in \supp_i$, we use $p_i(x)$ as a shorthand notation for $\Pr[X_i = x]$.
\begin{lemma} \label{lm:recursion}
Let $\gamma_i : \supp_i \rightarrow \mathbb{R}_{\geq 0}$ denote a function.
The following recurrence holds:
\begin{equation}
\label{eq:rec1}
\tau(i,a) = \min_{\gamma_i} \max_{x \in \supp_i}\{\tau(i-1, \gamma_i(x))+x\}, \text{ for any } i\in \{1,2,...,n\}
\text{ and } a\geq 0,
\end{equation}
where the minimum is taken among all functions $\gamma_i$ such that
\begin{equation} \label{eq:subject}
\sum_{x \in \supp_i}{\gamma_i(x) p_i(x)} \geq a.
\end{equation}

\end{lemma}

\noindent
Intuitively, $\gamma_i(x)$ represents the value $\Pr[\sum_{j=1}^{i-1}X_j\leq \tau(i,a)-x]$.
By enforcing \eqref{eq:subject}, we make sure that
$\Pr[\sum\limits_{j=1}^{i}{X_j} \leq \tau(i,a)] \geq a$ so that the definition of $\tau(i,a)$ is met.
The formal proof is as follows.

\begin{proof}
We first fix some $0< a\leq 1$ and prove that for any $\gamma_i: \supp_i \rightarrow [0,1]$,
the quantity
$C' := \max_{x\in\supp_i}\{\tau(i-1, \gamma_i(x))+x\} \geq \tau(i,a)$.
Due to the independence, we can see that
\begin{equation} \label{eq:independence}
\Pr\Bigl[\,\sum_{j = 1}^{i}{X_j} \leq C'\,\bigr] =
\sum_{x \in \supp_i} p_i(x) \cdot \Pr\Bigl[\,\sum_{j=1}^{i-1}{X_{j}} \leq C'-x\,\Bigr] .
\end{equation}
By the definition of $C'$,
we know that $C'-x \geq \tau(i-1,\gamma_i(x))$ for all $x\in \supp_i$,
and hence $\Pr[\sum_{j=1}^{i-1}{X_{j}} \leq C'-x] \geq \gamma_i(x)$. This gives
$
\Pr\Bigl[\,\sum_{j = 1}^{i}{X_j} \leq C'\,\Bigr] \geq \sum_{x \in \supp_i}\gamma_i(x) p_i(x) \geq a.
$
Therefore, we have $C' \geq \tau(i,a)$.
If $a > 1$, by (\ref{eq:subject}) we know that $\gamma_i(x') > 1$ for some $x' \in \supp_i$ and therefore $C' \geq \tau(i-1, \gamma_i(x')) = \infty$.
If $a=0$, $\tau(i,a)$ is defined to be $-\infty$ and we trivially have $C'\geq \tau(i,a)$.
In all above cases, we have $C' \geq \tau(i,a)$.

Now, we prove the reverse direction.
In particular, we prove that there exists a choice of $\gamma_i := \gamma_i^\star$
such that $C^\star := \max_{x \in \supp_i}\{\tau(i-1, \gamma_i^\star(x))+x\}\leq \tau(i,a)$.
First assume that $0< a\leq 1$.
Let
$
\gamma_i^\star(x) = \Pr\Bigl[\,\sum_{j=1}^{i-1}{X_j} \leq \tau(i,a)-x\,\Bigr].
$
It is easy to see that $\gamma_i^\star$ satisfies \eqref{eq:subject}, if we let $C' = \tau(i,a)$ in (\ref{eq:independence}) and use the definition of $\tau(i,a)$ on the LHS.
Therefore, $\tau(i-1, \gamma_i^\star(x)) \leq \tau(i,a)-x$ for all $x\in \supp_i$,
which leads to $\max_{x \in \supp_i}\{\tau(i-1, \gamma_i^\star(x))+x\} \leq \tau(i, a)$.
If $a=0$, let $\gamma_i^\star(x)=0$ for all $x\in \supp_i$. It is easy to see that
$C^\star\leq \tau(i,a)$ since $C^\star=-\infty$ by the definition of $C^\star$.
If $a> 1$, $\tau(i,a)=+\infty$ and $C^\star\leq \tau(i,a)$ is also trivially true.
This completes the proof.
\end{proof}

For ease of notation,
we define function $\lambda_i: \supp_i \rightarrow \mathbb{R}_{\geq 0}$ as $\lambda_i(x) = \gamma_i(x) p_i(x)/a$.
Hence, the recursion in Lemma \ref{lm:recursion} can be rewritten as:
\begin{equation} \label{eq:DP}
\tau(i,a) = \min_{\lambda_i} \max_{x \in \supp_i}\{\tau(i-1, \frac{\lambda_i(x)}{p_i(x)} a)+x\}
\quad\text{ subject to }\quad
\sum_{x \in \supp_i}{\lambda_i(x)} \geq 1.
\end{equation}

Since the second argument $a$ is a continuous variable, it is not clear how recursion~\eqref{eq:DP} can be efficiently evaluated.
To overcome this issue, we discretize the above recursion as follows:
Let
$$Q = 1+\frac{\epsilon}{n}.$$
Note that $Q$ is slightly larger than $1$.
Recall that $\q=\prod_{i} \Pr[X_i = 0]$, which is a trivial lower bound on $F(C)$.
Let $s = \lceil\log_Q{\frac{1}{\q}}\rceil = O(\frac{n}{\epsilon} \log {\frac{1}{\q}})$.

We define recursively the function $T: \{1,2,...,n\} \times \{0,1,2,...,s\} \rightarrow \Z \cup \{\pm \infty\}$ as follows:
The base case is defined as
\begin{equation}
T(1,j) = \tau(1,Q^{-j}), \text{ for all } j \in [s].
\end{equation}
Note that each $\tau(1,Q^{-j})$ can be computed with $O(\log C)$ calls to $\oracle_1$.
For $i \geq 2$ and $j \in [s]$, define
\begin{equation}
\label{eq: T-recurrence}
T(i,j) = \min_{\lambda_i}\max_{x \in \supp_i} \Bigl\{ T\Bigl(i-1, j+\lfloor\log_{Q}\frac{p_i(x)}{\lambda_i(x)}\rfloor\Bigr)+x \Bigr\}
\quad\text{ subject to }\quad
\sum_{x \in \supp_i}{\lambda_i(x)} \geq 1.
\end{equation}
To make analysis easier (mainly to avoid tedious case-by-case study near the boundary), we extend the domain of $T(i,j)$ to $\{1,2,...,n\} \times \Z \cup \{+\infty\}$ such that for all $i\in \{1,2,...,n\}$, we define $T(i,j)=0$ for $j>s$ and $T(i,j) = +\infty$ for $j < 0$.
Furthermore, we define $T(i, +\infty) = -\infty$. So if $\lambda_i(x)=0$, $T\Bigl(i-1, j+\lfloor\log_{Q}(\frac{p_i(x)}{\lambda_i(x)})\rfloor\Bigr)=T(i-1,+\infty)=-\infty$.
Comparing \eqref{eq: T-recurrence} and \eqref{eq:DP}, the similarity suggests
that $T(i,j)$ is an approximate version of $\tau(i, Q^{-j})$.
The next lemma formalizes this idea.

\eat{
\begin{algorithm}[t]
\caption{The Dynamic Program for Computing $\Pr[\sum_{i=1}^n X_i\leq C]$}
\label{algo1}
\begin{algorithmic}[1]
%
\STATE Set $Q = 1+\frac{\ln(1+\epsilon)}{n+1}$ and $s = \lceil\log_{Q}{\frac{1}{q}}\rceil$.
\STATE Query oracle $\mathcal{O}_1$ to obtain $T(1, j) = \tau(1, Q^{-j})$ for all $j \in [s]$.
\FOR{ $i = 2 \to n$}
\FOR{ $j = 0 \to s$}
\STATE Set $T(i,j) = \min_{\lambda_i}\max_{x \in \supp_i}{T(i-1, j+\lfloor\log_{Q}(\frac{p_i(x)}{\lambda_i(x)})\rfloor)+x}$.
\ENDFOR
\ENDFOR
\STATE Let $j^\star = \min\{j: T[n,j] \leq C\}$. Output $\tF(C) := Q^{-j^\star+1}$.
\end{algorithmic}
\end{algorithm}
}

\begin{algorithm}[t]
\caption{The Dynamic Program for Computing (\ref{eq:countprob})}
\label{algo1}

Let $Q = 1+\epsilon/n$ and $s = \lceil\log_{Q}{\frac{1}{\q}}\rceil$\;
Query oracle $\mathcal{O}_1$ to obtain $T(1, j) = \tau(1, Q^{-j})$ for all $j \in [s]$\;
\For{ $i = 2 \to n$}{
\For{ $j = 0 \to s$}{
    $T(i,j) \leftarrow \min_{\lambda_i}\max_{x \in \supp_i}{T(i-1, j+\lfloor\log_{Q}\frac{p_i(x)}{\lambda_i(x)}\rfloor)+x,}~~$  subject to $\sum_{x \in \supp_i}{\lambda_i(x)} \geq 1.$
    }}
$j^\star \leftarrow \min\{j: T(n,j) \leq C\}$. {\bf Return} $\tF(C) := Q^{-j^\star+1}$\;
\end{algorithm}

\begin{lemma} \label{lemma: approx}
For all $i \in \{1,2,..., n\}$ and $j \in \Z \cup \{\infty \}$, we have that
\begin{equation*}
\tau(i, Q^{-j}) \leq T(i,j) \leq \tau(i, Q^{-(j-i)}).
\end{equation*}
\end{lemma}
\begin{proof}
We prove this by induction.
The base case is trivially true by the definition of $T(1,j)$, even for $j<0$ and $j>s$.
Now assume that the statement is true for $i-1$ $(i\geq 2)$ and all $j \in  \Z \cup \{+\infty\}$.
We prove it is also true for $i$ and all $j \in  \Z \cup \{+\infty\}$.
By the induction hypothesis, we have that
$$
T\Bigl(\,i-1, \bigl\lfloor j+\log_Q \frac{p_i(x)}{\lambda_i(x)} \bigr\rfloor\,\Bigr)
\leq \tau\Bigl(\,i-1, Q^{-(\lfloor j+\log_Q \frac{p_i(x)}{\lambda_i(x)} \rfloor - i+1)  }\,\Bigr) \leq
\tau\Bigl(\,i-1, \frac{\lambda_i(x)}{p_i(x)}Q^{-(j-i)}\,\Bigr) \quad \text{ and }
$$
$$
T\Bigl(\,i-1, \bigl\lfloor j+ \log_Q \frac{p_i(x)}{\lambda_i(x)} \bigr\rfloor\,\Bigr)
\geq \tau\Bigl(\,i-1, Q^{-\lfloor j+ \log_Q \frac{p_i(x)}{\lambda_i(x)} \rfloor}\,\Bigr)
\geq \tau\Bigl(i-1, \frac{\lambda_i(x)}{p_i(x)} Q^{-j}\Bigr).
$$
Note that the above inequalities hold even $ j+\log_Q \frac{p_i(x)}{\lambda_i(x)}$ is negative
or larger than $s$, or $\lambda_i(x)=0$.
Taking the maximum over $x$ and the minimum over $\lambda_i$ does not change the direction of the inequalities.
Combining the above inequalities with \eqref{eq:DP}, we  complete the proof.
\end{proof}

With this lemma, we can approximate the recursion~\eqref{eq:DP} by solving the recursion~\eqref{eq: T-recurrence}.
The pseudocode of the dynamic program is provided in Algorithm~\ref{algo1}.
Note that it is still not clear how to implement \eqref{eq: T-recurrence} in polynomial time, as the number of possible functions $\lambda_i$ can be infinite.
We address this issue in the next section.
Assuming~\eqref{eq: T-recurrence} can be implemented efficiently,
we can show the output of the algorithm is a good approximation of the true probability with the following lemma.

\begin{lemma}
The output $\tF(C)$ of Algorithm~\ref{algo1} is a $(1\pm\epsilon)$-approximation of $F(C)$.
\end{lemma}
\begin{proof}
From the choice of $j^\star$, we know that $T(n,j^\star) \leq C < T(n,j^\star-1)$.
According to Lemma \ref{lemma: approx},
we have $\tau(n,Q^{-j^\star}) \leq C < \tau(n,Q^{n-j^\star+1})$. Therefore, $F(C) = \Pr[\sum_{i=1}^{n}{X_i} \leq C] \in [Q^{-j^\star},Q^{n-j^\star+1}]$.
By outputting $\tF(C) = Q^{-j^\star+1}$,
the approximation ratio can bounded as
$1-\epsilon \leq Q^{-n} \leq \tF(C)/F(C) \leq Q \leq  1+\epsilon$.
\end{proof}

\subsection{An Efficient Implementation using Binary Search}
\label{subsec:imp}

\eat{
\begin{algorithm}[t]
\caption{Efficient Implementation of Recursion~\eqref{eq: T-recurrence}}
\label{algo2}
\begin{algorithmic}[1]
\STATE Set $\mathcal{L} = 0$, $\mathcal{R} = C$.
\WHILE{$\mathcal{R} > \mathcal{L}$}
\STATE Set $\calT = \lfloor (\mathcal{L}+\mathcal{R}) / 2\rfloor$.
\STATE Let $P_{i}(0) = \Pr[X_i\in (\calT-\infty, \calT-T(i-1,0)]$
\FOR{$m = 1 \to s$}
\STATE Let $P_{i}(m)= \Pr[X_i\in (\calT-T(i-1,m-1), \calT-T(i-1,m)]$ (via oracle $\oracle_i$).
\ENDFOR
\IF{$\sum_{m=0}^{s}{Q^{j-m} P_i(m)} \geq 1$}
\STATE Set $\mathcal{R} = \calT$
\ELSE
\STATE Set $\mathcal{L} = \calT+1$.
\ENDIF
\ENDWHILE
\STATE Set $T(i,j) = \mathcal{L}$.
\end{algorithmic}
\end{algorithm}
}

\begin{algorithm}[t]
\caption{Efficient Implementation of Recursion~\eqref{eq: T-recurrence}}
\label{algo2}
Initially, let $\mathcal{L} = 0$, $\mathcal{R} = n(C+1)$\;
\While{$\mathcal{R} > \mathcal{L}$}{
    $\calT \leftarrow \lfloor (\mathcal{L}+\mathcal{R}) / 2\rfloor$\;
    $P_{i}(0) \leftarrow \Pr\bigl[X_i\in (\calT-\infty, \calT-T(i-1,0)]\,\bigr]$ (via oracle $\oracle_i$)\;
    \For{$m = 1,2,\ldots, s+1$}{
        $P_{i}(m)\leftarrow \Pr\bigl[X_i\in (\calT-T(i-1,m-1), \calT-T(i-1,m)]\,\bigr]$ (via oracle $\oracle_i$)\;
    }
   \eat{$Z=\sum_{m=0}^{s+1}{Q^{j-m} P_i(m)}$\;}
   $Z=\sum_{m=0}^{s+1}{Q^{j+s+1-m} P_i(m)}$\;
    \If{$Z \geq Q^{s+1}$}{ \label{line:cp}
        $\mathcal{R} \leftarrow \calT$\; \label{line:r}
    }\Else{
        $\mathcal{L} \leftarrow \calT+1$\; \label{line:l}
    }
}
{\bf Return} $T(i,j) := \mathcal{L}$\;
\end{algorithm}

In this subsection, we show how to implement the recursion \eqref{eq: T-recurrence} in polynomial time.
Suppose we have already computed $T(i',j')$ for all $i' < i$ and $0\leq j'\leq s$ and
we are trying to compute the value $T(i,j)$.
Our approach is based on a binary search over the range of $T(i,j)$.
We maintain an interval $[\calL,\calR]$ such that $T(i,j)\in [\calL,\calR]$ throughout the algorithm.
In each iteration, we make a guess $\calT=\lfloor (\calL+\calR)/2 \rfloor$ and decide whether $T(i,j) \leq \calT$ using
the criterion in Lemma~\ref{lm:binary}.
It is not immediately clear how we can check efficiently whether the criterion is met.
We address this issue in Lemma~\ref{lm:segment}.
The pseudocode is provided in Algorithm~\ref{algo2}.

\begin{lemma} \label{lm:binary}
Suppose we are computing $T(i,j)$ and $\calT$ is the current guess.
We define that for all  $x \in \supp_i$,
$$\lambda_i'(x) = \max\bigl\{z\ \mid z \in \R_{\geq 0} \text{ and } T\bigl(i-1, j+\lfloor\log_{Q}\frac{p_i(x)}{z}\rfloor\bigr)+x \leq \calT\bigr\}. $$
Then,
$T(i,j) \leq \calT$ if and only if $\sum_{x \in \supp_i}\lambda_i'(x) \geq 1$.
\end{lemma}
\begin{proof}
Suppose $\sum_{x \in \supp_i}\lambda_i'(x) \geq 1$.
This means that there exists a choice of $\lambda_i := \lambda_i'$ such that
$$
\max_{x \in \supp_i}\Bigl\{T(i-1, j+\lfloor\log_{Q}\frac{p_i(x)}{\lambda_i(x)}\rfloor)+x\Bigr\} \leq \calT,
$$
and therefore $T(i,j) \leq \calT$ by \eqref{eq: T-recurrence}.

On the other hand, assume $T(i,j) \leq \calT$.
Suppose that $T(i,j)$ is achieved by $\lambda_i := \lambda_i^\star$ in \eqref{eq: T-recurrence}.
We know that $T(i-1,j+\lfloor \log_Q\frac{p_i(x)}{\lambda^\star_i(x)} \rfloor )+x \leq T(i,j)\leq \calT$ for any $x\in \supp_i$.
By definition of $\lambda_i'$, we see that $\lambda_i^\star(x) \leq \lambda'_i(x)$, for all $x \in \supp_i$. Therefore, $\sum_{x \in \supp_i}{\lambda_i'(x)} \geq \sum_{x \in \supp_i}{\lambda_i^\star(x)} \geq 1$.
\end{proof}

From Lemma~\ref{lm:binary}, we can see that if $|\supp_i|$ is bounded by a polynomial,
we can decide whether $T(i,j)\leq \calT$ in polynomial time.
However, $|\supp_i|$ can be exponentially large.
We resolve this issue by showing that $\sum_{x\in\supp_i}\lambda'_i(x)$ can be computed using at most $O(s)$
calls to $\oracle_i$.
For this purpose, we divide the support of $X_i$
into $s+3$ segments such that the sum of $\lambda_i'(x)$ values over each segment can be computed efficiently.
In particular, we divide $\N$ into segments
\begin{align*}
\Sg_0 &=(-\infty, \calT-T(i-1,0)] \cap \N, \\
\Sg_k &=(\calT-T(i-1,k-1), \calT-T(i-1,k)] \cap \N, \text{ for all } k=1,\ldots, s+1, \\
\Sg_{s+2} &=( \calT-T(i-1,s+1),\infty) \cap \N.
\end{align*}

Let us first investigate the last segment.
\begin{lemma}
\label{lm:segment}
Suppose $x \in \Sg_{s+2}$
and $\lambda_i'(x)$ is defined as in Lemma~\ref{lm:binary}. Then, we have that
$\lambda_i'(x) = 0.$
\end{lemma}
\begin{proof}
First, recall that $T(i-1,s+1)=0$.
For any $x>\calT$, the only possible nonnegative value $z$
that can make
$T\bigl(i-1, j+\lfloor\log_{Q}(\frac{p_i(x)}{z})\rfloor\bigr)+x \leq \calT$
is $0$, since $T(i-1,+\infty)=-\infty$.
\end{proof}

For each other segment,
$\lambda_i'(x)$ can be computed using the following lemma.

\begin{lemma}
\label{lm:segment}
Suppose $x \in \Sg_m$ for some $m \in [0,s+1]$
and $\lambda_i'(x)$ is defined as in Lemma~\ref{lm:binary}. Then, we have that
\begin{equation}
\lambda_i'(x) = Q^{j-m} p_i(x).
\end{equation}
\end{lemma}
\begin{proof}
Fix any  $x \in \Sg_m$ for some $m \in [0,s+1]$.
Let $z>0$ be such that
\begin{equation*}
j+\log_Q\Bigl(\frac{p_i(x)}{z}\Bigr) = m.
\end{equation*}
Now, we show $\lambda'_i(x)=z$.
Since $x\in \Sg_m$, we can see that $\calT-x\in [T(i-1,m),T(i-1,m-1))$.
Therefore, we have that
$$
T\bigl(i-1,j+\lfloor \log_Q\Bigl(\frac{p_i(x)}{z}\Bigr) \rfloor\bigr) = T(i-1,m) \leq \calT-x.
$$
Now, we show any value $z'$ larger than $z$ does not satisfy the above inequality.
This is because
$T(i-1,j+\lfloor \log_Q\Bigl(\frac{p_i(x)}{z'}\Bigr) \rfloor)$ is at least $T(i-1,m-1)$
which is greater than $\calT-x$.
Since $z$ is the largest value that makes the inequality true, we have $\lambda_i'(x)=z$.
\end{proof}

As a result, it is sufficient to ask the $\oracle_i$ for the probability values $\strin{P_i(m) = }\Pr[X_i\in \Sg_m]$
and we can compute $\sum_{x\in \Sg_m}\lambda_i'(x)$ simply by
\begin{equation} \label{eq:sum_of_lambda}
\sum_{x\in \Sg_m}\lambda_i'(x)= Q^{j-m} P_i(m).
\end{equation}

Together with Lemma~\ref{lm:binary}, (\ref{eq:sum_of_lambda}) shows that the criterion $T(i,j) \leq \calT$ is equivalent to $\sum_{m=0}^{s+1}$ ${Q^{j-m} P_i(m)} \geq 1$ . Multiplying $Q^{s+1}$ on both sides, we derive Line~\ref{line:cp} of Algorithm~\ref{algo2}, which is used to decide the criterion $T(i,j) \leq \calT $.

\begin{lemma}
\label{eq: cost_trans}
Suppose the oracle can answer each query in constant time.
The optimization over all $\lambda_i$ and $x$ in equation (\ref{eq: T-recurrence}) can be done in time
$O\bigl(\frac{n}{\epsilon}\log{\frac{1}{\q}}\log C\bigr)$,
where $\q=\prod_{i} \Pr[X_i=0]$.
\end{lemma}
\begin{proof}
The binary search takes $\log C$ iterations, while in each iteration we make at most $O(s)$ oracle calls.
So the implementation of the recursion step runs in $O(s \log C) = O(\frac{n}{\epsilon}\log \frac{1}{\q} \log C)$
since $s = O(\frac{n}{\epsilon} \log{\frac{1}{\q}})$.
\end{proof}

The number of states in the dynamic program is $O(n s)$.
Combined with Lemma \ref{eq: cost_trans},
the total running time of Algorithm \ref{algo1}, which calls Algorithm \ref{algo2} as a subroutine, is $O\bigl(\frac{n^3}{\epsilon^2} \log(\frac{1}{\q})^2 \log C\bigr)$.
This completes the proof of Theorem \ref{th:main}.

\subsection{The Bit Complexity of the Algorithm}
\label{sec:bit}
Our results have been stated under the assumption that any real arithmetic operation has unit cost.
Now we consider the computation model
where only bit operations are allowed.
Moreover, the oracles also return numerical values that are encoded in bits.
To model this, we assume that each query to the oracle contains an additional integer $\ell$ (called precision paramter),
which specifies the number of bits the oracle should return.
Upon such a query, the oracle returns the first (i.e., the most significant) $\ell$ bits of the answer, which is always between 0 and 1
\footnote{
This can model the case where the oracle is given by a mathematical formula
and there is a numerical algorithm that can compute the answer to arbitrary precision.
}.
Our goal is to show that a slightly modified algorithm still preserves a ($1\pm\epsilon$)-approximation, while the bit complexity is also
$\poly(n, \epsilon^{-1}, \log C, \log \frac{1}{\q})$.

First, we use the oracle to figure out a close estimate of $\log 1/\q$,
since it determines $s$, which relates to the table size of the dynamic problem.
For each $\oracle_i$, we repeatedly query the value
$\Pr[X_i=0]$ with precision $\ell=1,2,3,\ldots$ until we get the first nonzero bit.
It is easy to see that we stop with $\ell_i = \lceil -\log \Pr[X_i=0]\rceil+1$ steps.
It then suffices for our algorithm to use an upper bound of $\log 1/\q$,
which we choose to be $\sum_i \ell_i$.

\eat{In the sequel, we run Algorithm~\ref{algo1}, which calls Algorithm~\ref{algo2} as a subroutine.
We use a slightly different value of $Q$, which we choose to be $1+2^{-\ell_q}$, where $\ell_q = \lceil \log_2 \frac{2n}{\epsilon} \rceil$.
Notice that $\calL$, $\calR$, $\calT$ can be encoded in $O(\log C)$ bits, and $Q$ in $O(\log \frac{n}{\epsilon})$ bits. When we multiply a number of $a$ bits and a number of $b$ bits, we use $a+b$ bits for the result so that no precision is lost in mulplication. Hence, any $Q^{j-m} (0\leq m\leq s)$ used in the algorithm can be encoded in $O(s \log \frac{n}{\epsilon})$ bits.}

In the sequel, we run Algorithm~\ref{algo1} and Algorithm~\ref{algo2} as its subroutine with a slightly different value of $Q = 1+2^{-\ell_q}$, where $\ell_q = \lceil \log_2 \frac{n}{\epsilon} \rceil+2$, so that $Q \leq 1+\frac{\eps}{4n}$. Notice that  $\calL$, $\calR$, $\calT$ can be encoded in $O(\log C)$ bits and with $O(\log \frac{n}{\epsilon})$ bits, $Q$ has exact representation. When we multiply a number of $a$ bits and a number of $b$ bits, we use $a+b$ bits for the result so that no precision is lost. Hence, any $\done{Q^{j+s+1-m}} (0\leq m\leq \done{s+1})$ used in the algorithm can be encoded in $O(s \log \frac{n}{\epsilon})$ bits.

\eat{However, there is no guarantee that $Q^{j-m}$ has a finite exact representation when $m > j$. To resolve the issue, we make a mild modification to Line \ref{line:cp} of Algorithm~\ref{algo2}: instead of deciding whether $Z \geq 1$, we compare $Y = \sum_{m=0}^{s+1}{Q^{j} P_i(m)} \geq Q^{m}$. }

The problematic part is the answer $P_i(m)$ returned by the oracle (in Steps 4 and 7 of Algorithm~\ref{algo2}) is of finite precision.
Whenever Algorithm~\ref{algo2} attempts to make an oracle call, e.g., $\Pr[n_1\leq X_i\leq n_2]$, we set the precision parameter to be $L$,
which we will determine later, and call $\oracle_i$, which returns a truncated probability $\widehat{P}_i(m)$ with a truncation error $P_i(m)-\widehat{P}_i(m) \in [0, 2^{\done{-L+1}})$.
To distinguish from the original version of Algorithm~\ref{algo2},
we call the new version with finite precision {\em the bit version}.
We use $\widehat{T}(i,j)$ to
denote the $T(i,j)$ value obtained from the bit version of Algorithm~\ref{algo2}.
To account for this truncation error, we need a modified version of Lemma \ref{lemma: approx}:

\begin{lemma} \label{lm:approx_finite}
\done{Let the precision parameter be $L$ for each oracle call.} For all $i \in \{1,2,...,n\}$ and $j \in \mathbb{Z} \cup \{+\infty\}$, we have
\begin{equation} \label{eq:hat_t}
\tau(i,Q^{-j}) \leq \hT(i,j) \leq \tau(i, (1+\eta)^{\done{i}} Q^{-(j-i)} ),
\end{equation}
where \done{$\eta = (Q^{s+1}-Q^{-1})/(Q-1) \cdot 2^{-(L-1)}$}.
\end{lemma}

\begin{proof}
The statement is trivial for \done{$j > s$ and $j < 0$}. 
For $j \in [s]$, we prove the lemma by induction on $i$.

We first prove the induction step.
The proof for the base case $i = 1$ is very similar and we present it at the end. 
For $i \geq 2$, 
suppose in the dynamic program, we have computed
the $\hT(i-1,j), ~ \forall j\in \Z$.
Assume the statement of the lemma is true for $i-1 ~(i \geq 2)$ and $j \in \Z$. Now, we prove it for $i$ and $j \in \Z$.
Fix a particular $0\leq j\leq s$.
Imagine that we run two copies of Algorithm~\ref{algo2} simultaneously:
one copy $\mathcal{A}_1$ is the original version with infinite precision, and the other $\mathcal{A}$ is the bit version
\footnote{
We note that both $\calA_1$ and $\calA$ use $\hT(i-1,j)$ values.
So the output of $\calA_1$ may not be the exact value of $T(i,j)$.
}.
Clearly, the execution flows of both $\mathcal{A}_1$ and $\mathcal{A}$ would be the same in the beginning (i.e., $\calL,\calR,\calT$ values are the same).
If the execution flows are the same throughout, the outputs of $\calA_1$ and $\calA$ are the same.
Otherwise, we consider the snapshot when the deviation first occurs
and the condition in Line \ref{line:cp} returns different boolean values between $\mathcal{A}_1$ and $\mathcal{A}$.
Recall that \done{$\K = \sum_{m = 0}^{s+1}{Q^{j+s+1-m} P_i(m)}$} and let \done{$\hK = \sum_{m = 0}^{s+1}{Q^{j+s+1-m} \hP_i(m)}$}.
It is easy to see that \done{$\hK < Q^{s+1} \leq \K$}.
So the only possibility is that  $\mathcal{A}_1$ jumps to Line \ref{line:r} and $\mathcal{A}$ goes to Line \ref{line:l}.
Let $T_1$ be output of $\calA_1$.
The above argument shows that $\hT(i,j) \geq T_1$.

Now we show $T_1 \geq \tau(i, Q^{-j})$.
By \eqref{eq: T-recurrence} and Lemma~\ref{lm:binary} (with $T(i-1,\cdot)$ replaced by $\hT(i-1,\cdot)$), $T_1$ is the solution of the following optimization problem:
\begin{equation*}
T_1 = \min\limits_{\lambda_i}\max\limits_{x \in \supp_i}\Big\{\hT(i-1, j+\lfloor \log_Q(\frac{p_i(x)}{\lambda_i(x)})\rfloor)+x\Big\}, ~~ \text{subject to } \sum\limits_{x \in \supp_i}{\lambda_i(x)} \geq 1,
\end{equation*}
Using the same proof as in Lemma \ref{lemma: approx}, we can show that $T_1 \geq \tau(i, Q^{-j})$.
Hence, we have one direction of the lemma:
$\hT(i,j) \geq \tau(i, Q^{-j}).$

The proof for the other direction is similar.
This time, we run $\mathcal{A}$ simultaneously with a slightly modified version $\calA_2$ of Algorithm~\ref{algo2}.
Suppose that the oracles in $\calA_2$ can return values with infinite precision.
The goal of $\mathcal{A}_2$ is to solve the following optimization problem:
\begin{equation}
\label{opt2}
T_2 = \min\limits_{\lambda_i}\max\limits_{x \in \supp_i}\Big\{\hT(i-1, j+\lfloor \log_Q(\frac{p_i(x)}{\lambda_i(x)})\rfloor)+x\Big\}, ~~ \text{subject to } \sum\limits_{x \in \supp_i}{\lambda_i(x)} \geq 1+\eta.
\end{equation}
We immediately can see that $T_2$ is an approximate value of $\tilde{\tau}(i, Q^{-(j-i)})$, which is defined as
\begin{equation} \label{eq:tau_tilde}
\tilde{\tau}(i, Q^{-(j-i)}) = \min\limits_{\lambda_i}\max\limits_{x \in \supp_i}\Big\{\tau(i-1, \frac{\lambda_i(x)}{p_i(x)} Q^{-(j-i)})+x\Big\}, ~~\text{subject to } \sum\limits_{x \in \supp_i}{\lambda_i(x)} \geq 1+\eta.
\end{equation}
Using the same argument as in Lemma \ref{lemma: approx}, we can show that
\begin{equation}  \label{eq:t2}
T_2 \leq \tilde{\tau}(i, (1+\eta)^{i-1} Q^{-(j-i)} ).
\end{equation}
Furthermore, comparing (\ref{eq:tau_tilde}) with (\ref{eq:DP}), we can see the connection between $\tau$ and $\tilde{\tau}$ as follows:
\begin{equation} \label{eq:tau_conn}
\tilde{\tau}(i, (1+\eta)^{i-1} Q^{-(j-i)} ) = \tau(i, (1+\eta)^{i} Q^{-(j-i)}).
 \end{equation}
Using the same argument as in Section~\ref{subsec:imp},
the optimization problem~\eqref{opt2} can be implemented by a binary search procedure almost identical to Algorithm~\ref{algo2}
except that $\calA_2$ compares $\K$ with $(1+\eta)Q^{s+1}$ in Line~\ref{line:cp}.
Notice that $\eta$ has been chosen carefully such that the inequality $\K \geq (1+\eta)Q^{s+1}$ implies $\hK \geq Q^{s+1}$.
Using a similar execution-flow analysis,
we can see that the deviation occurs when \done{$Q^{s+1} \leq \hK<\K<(1+\eta) Q^{s+1}$} and consequently we have that $T_2 \geq \hT(i,j)$.
By (\ref{eq:t2}) and (\ref{eq:tau_conn}), we have that
$\hT_2 \leq \tau(i, (1+\eta)^{i-1} Q^{-(j-i)}).$

Now we prove the base case $i = 1$.
The proof is very similar to the induction step. 
Recall that in Algorithm~\ref{algo1}, we use binary search to compute the value $T(1,j)$ exactly. The binary search maintains an interval $[\calL, \calR]$ such that $T(1,j) \in [\calL, \calR]$ throughout. 
In each iteration, we make a guess $\T$, compare it with $Q^{-j}$, and update $\calL \leftarrow \T+1$ if $\Pr[X_1 \leq \T] < Q^{-j}$, 
or $\calR \leftarrow \calT$ otherwise. 
In the bit version, $Q^{-j}$ and $\Pr[X_1 \leq \T]$ may not be exact. Nevertheless, we can still get approximation $\hT(1,j)$, which is $T(1,j)$ of finite precision using a slightly modified binary search program $\mathcal{B}$. Denote the truncated probability from the oracle $\oracle_1$ as $\hatPr[X_1 \leq \T]$. $\mathcal{B}$ follows the previous binary search scheme but replaces the decision rule $\Pr[X_1 \leq \T] < Q^{-j}$
by $Q^{j} \cdot \hatPr[X_1 \leq \T] < 1$. 
We first show the output of $\mathcal{B}$, $\hT(1,j)$ is at least $\tau(1,Q^{-j})$. Imagine another copy of 
this binary search program $\mathcal{B}_1$ running simultaneously but with infinite precision. 
Obviously, $\mathcal{B}_1$ computes the exact value $T(1,j) = \tau(1,Q^{-j})$. 
If the execution flows of $\mathcal{B}$ and $\mathcal{B}_1$
are the same throughout, then the output would be the same: $\hT(1,j) = T(1,j)$. 
Otherwise, by considering the snapshot when the deviation first occurs, i.e.,
$Q^{j} \cdot \hatPr[X_1 \leq \T] < 1 \leq Q^{j} \cdot \hatPr[X_1 \leq \T]$. 
It is not hard to see that $\hT(1,j) \geq T(1,j) = \tau(1,Q^{-j})$. 
Next, we show the upper bound by proving $\hT(1,j) \leq \tau(1, (1+\eta)Q^{-j}) \leq \tau(1, (1+\eta)Q^{1-j})$. 
Imagine a third copy of this binary search program $\mathcal{B}_2$ running simultaneously, 
but with the decision rule replaced by $Q^{j} \cdot \hatPr[X_1 \leq \T] < 1+\eta$. 
By the choice of $\eta$, 
the same execution-flow analysis gives that 
$\hT(1,j) \leq \tau(1,(1+\eta) Q^{-j})$. This concludes the proof of the base case.
\end{proof}

\done{When all $\hT(n,j)$ values have been computed (for $j \in [0,s]$), the bit version of Algorithm~\ref{algo1} returns $Q^{-j^\star+1}$ up to precision $\ell_{\text{ans}}$. The following lemma shows the returned $\tilde{F}(C)$ is an ($1\pm\epsilon$)-approximation of $F(C)$.}

\begin{lemma} \label{lm:bit_final}
By choosing \done{$Q = 1+2^{-l_q} \leq 1+\frac{\epsilon}{4n}$ for $\ell_q = \lceil \log_2 \frac{n}{\epsilon} \rceil + 2$, $L \geq \log_2 \frac{32 n^2 Q}{\eps^2}+ s + 1$ and $\ell_{\text{ans}} = s+\log_2 \frac{Q}{\eps}$} \eat{$L = \Omega(\frac{n}{\epsilon}\log \frac{n}{\epsilon} \log \frac{1}{\q})$}, the above modified algorithm can produce a ($1\pm\epsilon$)-approximation of $F(C)$.
\end{lemma}

\begin{proof}
According to Lemma \ref{lm:approx_finite}, we have $\tau(n, Q^{-j^\star}) \leq C \leq \tau(n, Q^{n-j^\star+1} (1+\eta)^{n})$, and hence $F(C) \in [Q^{-j^\star}, Q^{n-j^\star+1} (1+\eta)^{n}]$. On the other hand, the returned $\tF(C)$ by our bit version of the algorithm satisfies $\tF(C) \in (Q^{-j^\star+1}-2^{-\ell_{\text{ans}}+1}, Q^{-j^\star+1}]$. By our choice of $Q$ (or $\ell_q$), $L$ and $\ell_{\text{ans}}$, it is not hard to verify that $\tF(C)/F(C) \in [1-\epsilon, 1+\epsilon]$.
\end{proof}

\if
Now we analyze the bit complexity of the modified algorithm by listing that of each operation as follows:

\vspace{15pt}
\begin{tabular}{lll}
Operation & Element & Bit Complexity \\ \hline
Multiply & $Q^{j-m}$, $P_i$ &  $O(\ell \cdot sn/\epsilon)$ \\
Sum & $Q^{j-m} \cdot P_i$ & $O(s \cdot \ell s n /\epsilon)$ \\
Compare & $\sum_{m=0}^s{Q^{j-m} P_i}$, $1$ & $O(s \cdot \ell s n /\epsilon)$ \\
Add/Sub & $T$, $\calT$ & $O(\log C)$
\end{tabular}
\vspace{15pt}
\fi

The required bit length $L$ (the length of $\widehat P_i(m)$ values) is $\poly(n,\epsilon^{-1},\log \frac{1}{\Delta})$.
So all numbers involved in our algorithms can be encoded in $\poly(n, \epsilon^{-1}, \log C, \log \frac{1}{\q})$ bits,
and the total number of arithmetic operations (additions, multiplications, comparisons) is also $\poly(n, \epsilon^{-1}, \log C, \log \frac{1}{\q})$,
we then conclude the overall bit complexity is $\poly(n, \epsilon^{-1}, \log C, \log \frac{1}{\q})$.


\noindent

\section*{Acknowledgements}
We would like to thank  Paul Tsui and the anonymous reviewer for their detailed comments,
which helped us to improve the presentation significantly.
This work was supported in part by
the National Basic Research Program of China Grant
2011CBA00300, 2011CBA00301, the National Natural Science Foundation of China Grant 61202009,
61033001, 61061130540, 61073174.

\bibliographystyle{abbrv}
\bibliography{countProb}

\end{spacing}
\end{document}